\documentclass[aps,twocolumn,pra,showpacs]{revtex4-1}

\usepackage{amsfonts,amsthm}
\usepackage{subfigure}
\usepackage{amssymb}
\usepackage{amscd}
\usepackage{amsmath}    
\usepackage{multirow}
\usepackage{graphicx}
\usepackage{dcolumn}
\usepackage{bm}

\newtheorem{lemma}{Lemma}
\newtheorem{theorem}{Theorem}

\renewcommand{\vec}[1]{\mathbf{#1}}

\begin{document}
\title{ Tight detection efficiency bounds of Bell tests in no-signaling theories}
\author{Zhu Cao}
\email{cao-z13@mails.tsinghua.edu.cn}
\author{Tianyi Peng}
\email{perfectty@qq.com}
\address{Center for Quantum Information, Institute for Interdisciplinary Information Sciences, Tsinghua University, Beijing, China}

\begin{abstract}
No-signaling theories, which can contain nonlocal correlations stronger than classical correlations but limited by the no-signaling condition, have deepened our understanding of the quantum theory. In principle, the nonlocality of these theories can be verified via Bell tests. In practice, however, inefficient detectors may make Bell tests unreliable, which is called the detection efficiency loophole. In this work, we show almost tight lower and upper bounds of the detector efficiency requirement for demonstrating the nonlocality of no-signaling theories, by designing a general class of Bell tests. In particular, we show tight bounds for two scenarios: the bipartite case and the multipartite case with a large number of parties. To some extent, these tight efficiency bounds quantify the nonlocality of no-signaling theories. Furthermore, our result shows that the detector efficiency can be arbitrarily low even for Bell tests with two parties, by increasing the number of measurement settings.
Our work also sheds light on the detector efficiency requirement for showing the nonlocality of the quantum theory.
\end{abstract}

\maketitle

\section{Introduction}
Quantum mechanics allows remote parties to be entangled in a way that is beyond classical physics. Bell tests, being an elegant illustration of this phenomenon, perform projective measurements on entangled parties who cannot signal to each other, for instance, enforced by space-like separation or a shield between the remote parties, to generate correlations of measurement outcomes that are impossible for the classical theory (i.e., local realism) \cite{bell}. By exploiting Bell tests, Popescu and Rohrlich generalize the theories beyond quantum mechanics, which are still restricted by the no-signaling condition \cite{prbox}. This immediately raises several interesting questions. One is to understand the restriction on quantum mechanics beyond the no-signaling theory, which inspires a line of works on quantum foundations \cite{Hardy2001quant}. Another is whether there exists a theory more nonlocal than the quantum theory in nature, such as the one only limited by the no-signaling condition (later referred to as \emph{the maximally nonlocal theory}). Posing constraints in addition to the no-signaling condition, such as the uncertainty principle, induces various no-signaling theories, that are less nonlocal than the maximally nonlocal theory. Though such theories have not been experimentally found, researchers postulate that they may exist in ultra-high-density objects such as black holes \cite{Preskill1992Do} or when the scale is out of the quantum regime, for example, smaller than a Planck length. Independently, this research has spurred interest for device-independent quantum information processing where the adversary is only limited to the no-signaling condition \cite{mayers1998quantum,QKDNoSignal05}.

One of the major obstacles of Bell tests
is the detection efficiency loophole. In fact, loophole-free Bell tests have only been very recently demonstrated \cite{Hensen2015Loophole, LoopholeFree:Zeilinger, LoopholeFree:NIST}. They are all based on Bell inequalities that have two measurement settings for two parties and thus need a detection efficiency bound of at least $2/3$ \cite{PhysRevA.68.062109}. It was known that $2/3$ is tight for two measurement settings \cite{Eberhard93}. However, beyond that, the efficiency bounds of no-signaling theories for more measurement settings and/or parties are largely unknown, even for the most-studied quantum theory. The first result in this direction is due to Larsson and Semitecolos \cite{larsson2001strict}, who show that, for $N$ parties and two measurement settings, the detector efficiency requirement is no larger than $N/(2N+1)$. This bound is later shown to be tight by Massar and Pironio \cite{PhysRevA.68.062109}. In the case of two parties, the best upper bound with four measurement settings is $61.8\%$ \cite{vertesi2010closing} and with a large number of measurement settings, the upper bound can approach zero \cite{massar2002nonlocality}. In the tripartite case, the best upper bound for eight measurement settings is $0.501$ \cite{pal2015closing}. For an infinite number of parties,
a series of works
\cite{PhysRevLett.91.047903,PhysRevA.86.062111,PhysRevA.92.052104}
have led to the best upper bound $2/(2+M)$ for $M$ measurement settings.

In practice, to lower the detection efficiency requirement, Bell tests with more than two measurement settings are more important than the ones with two measurement settings because they potentially have a lower efficiency requirement. Consequently they could make experimental realizations easier, especially for optical systems where the loss is normally high. The minimum detector efficiency to violate the local realism can also be viewed as a measure to characterize nonlocal correlations and thus is an important operational quantity.

\section{Preliminaries}
Before continuing, we first formalize the problem of Bell tests.
In Bell tests, there are two space-separated parties, $A$ and $B$. In each turn, $A$ ($B$) will be given a number $x \in [1,M_A]$ ($y \in [1,M_B]$)  randomly. Then $A$ and $B$ will output the outcomes $a$ and $b$ respectively according to the inputs and their shared resources.

Since the two parties cannot signal to each other in a Bell test, $A$ $(B)$ is unaware of the input to $B$ $(A)$. Thus a no-signaling probability $p^{NS}$ satisfies
\begin{equation}
\begin{aligned}
\forall x,y \quad p^{NS}(a|x,y)= p^{NS}(a|x), \\
p^{NS}(b|x,y) = p^{NS}(b|y),\\
\end{aligned}
\end{equation}
 which represents that the probability of one party $A$ $(B)$ outputting an outcome $a$ $(b)$ is independent of the input of the other party $B$ $(A)$ and only depends on its own input $x$ ($y$).

Denote the detector efficiency as $\eta$, which is independent of the input. The probability under such inefficiency of detectors, denoted as $p^{NS}_{\eta}$, is related to the ideal no-signaling probability $p^{NS}$ by
\begin{equation}
\begin{aligned}
p^{NS}_\eta(a,b|x,y) =& \eta^2p^{NS}(a,b|x,y),
\\p^{NS}_\eta(\Phi,\Phi) = &(1-\eta)^2,
\\
p^{NS}_\eta(a,\Phi|x) =& \eta(1-\eta)p^{NS}(a|x),\\
p^{NS}_\eta(\Phi,b|y) =& \eta(1-\eta)p^{NS}(b|y),
\end{aligned}
\end{equation}
where $\Phi$ denotes the empty output corresponding to a failed detection. The derivation of these relations is straightforward. For example, for the first relation, the outcomes $a$ and $b$ would be obtained only when both detectors respond. Hence the probability shrinks by a factor $\eta^2$ because the detector of each party responds with a probability $\eta$ and the detectors of different parties are independent.

A local hidden variable  (LHV) model assumes that $A$ and $B$ share a random variable $\lambda$ but cannot communicate. The strategy of $A$ ($B$) can then be characterized by the probability $p_A(a|x,\lambda)$ ($p_B(b|y, \lambda)$), which uses $\lambda$ and $x$ ($y$) to determine the probability of outputting $a$ ($b$). If $p^{NS}_{\eta}$ can be  simulated by a LHV model, there exists a local strategy such that for any outcomes $a$ and $b$ (including the empty output $\Phi$),
\begin{equation}
p^{NS}_{\eta} (a,b|x,y) = \int_{\lambda}p(\lambda)p(a|x,\lambda)p(b|y,\lambda)d\lambda.
\end{equation}
Since the efficiency $\eta$ and the simulation capability of a LHV model are monotonic (as shown in  the Appendix \ref{app}), one can define the minimum value of detector efficiency $\eta^*_{NS}$ for showing the nonlocality of the maximally nonlocal theory.

\section{Results}
We are now ready to state our main results.

\begin{theorem}
In Bell tests with two parties,
\begin{equation}
\eta_{NS}^{*}\ge \frac{M_A+M_B-2}{M_AM_B-1},
\end{equation}
 where $M_A,M_B$ are the numbers of inputs of $A$ and $B$ respectively.
\end{theorem}
\begin{proof}
Theorem 1 in Ref.~\cite{PhysRevA.68.062109} designs a LHV model to simulate any quantum strategy if the efficiency is lower than or equals $(M_A+M_B-2)/(M_AM_B-1)$. Since this design only leverages the no-signaling condition, the efficiency bound also holds in the maximally nonlocal theory which completes the proof.
\end{proof}

When $M_A=M_B=2$, Theorem 1 gives a well-known detector efficiency bound of $2/3$ which is tight for the quantum theory	\cite{larsson2001strict}. We now show that, for arbitrary  $M_A$ and $M_B$, the efficiency bound $(M_A+M_B-2)/(M_AM_B-1)$ is tight for the maximally nonlocal theory. The same statement is open for the quantum theory.

\begin{theorem}
In two-party Bell tests with $M_A$ and $M_B$ inputs respectively,
\begin{equation}
\label{eq:twoupper}
\eta_{NS}^{*}\le \frac{M_A+M_B-2}{M_AM_B-1}.
\end{equation}
\end{theorem}

\begin{proof}
Denote $P$ as the minimal prime number such that $P\geq \max(M_A,M_B) $. We construct a Bell inequality with $P$ outputs excluding $\Phi$,
\begin{equation}
\begin{aligned}
&I_{M_AM_B(P+1)(P+1)} = \sum_{a,b,x,y} f(a,b,x,y)p(a,b|x,y) \\
&\quad\quad+ \sum_{a,x} g(a,x)p(a|x) + \sum_{b,y} h(b,y)p(b|y),\\
&f(a,b,x,y) =
\begin{cases}
1 &  ((x>1) \vee (y>1)) \wedge (a+b \equiv xy)\\
0 &  (x=1) \wedge(y=1) \wedge (a+b \equiv xy) \\
-P^4 & \text{otherwise}
\end{cases}
\\
&g(a,x) =
\begin{cases}
-1 & x > 1\\
0 & x = 1
\end{cases},
\quad h(b,y) =
\begin{cases}
-1 & y > 1\\
0 & y = 1
\end{cases},
\end{aligned}
\end{equation}
where $1\leq x\leq M_A, 1\leq y \leq M_B, 0\leq a,b<P$ and the modulo is over $P$. An illustration of this Bell inequality is shown in Fig.~\ref{fig:2Bell}(a).
\begin{figure}
\centering
    \subfigure[]{%
      \includegraphics[width=4.2cm]{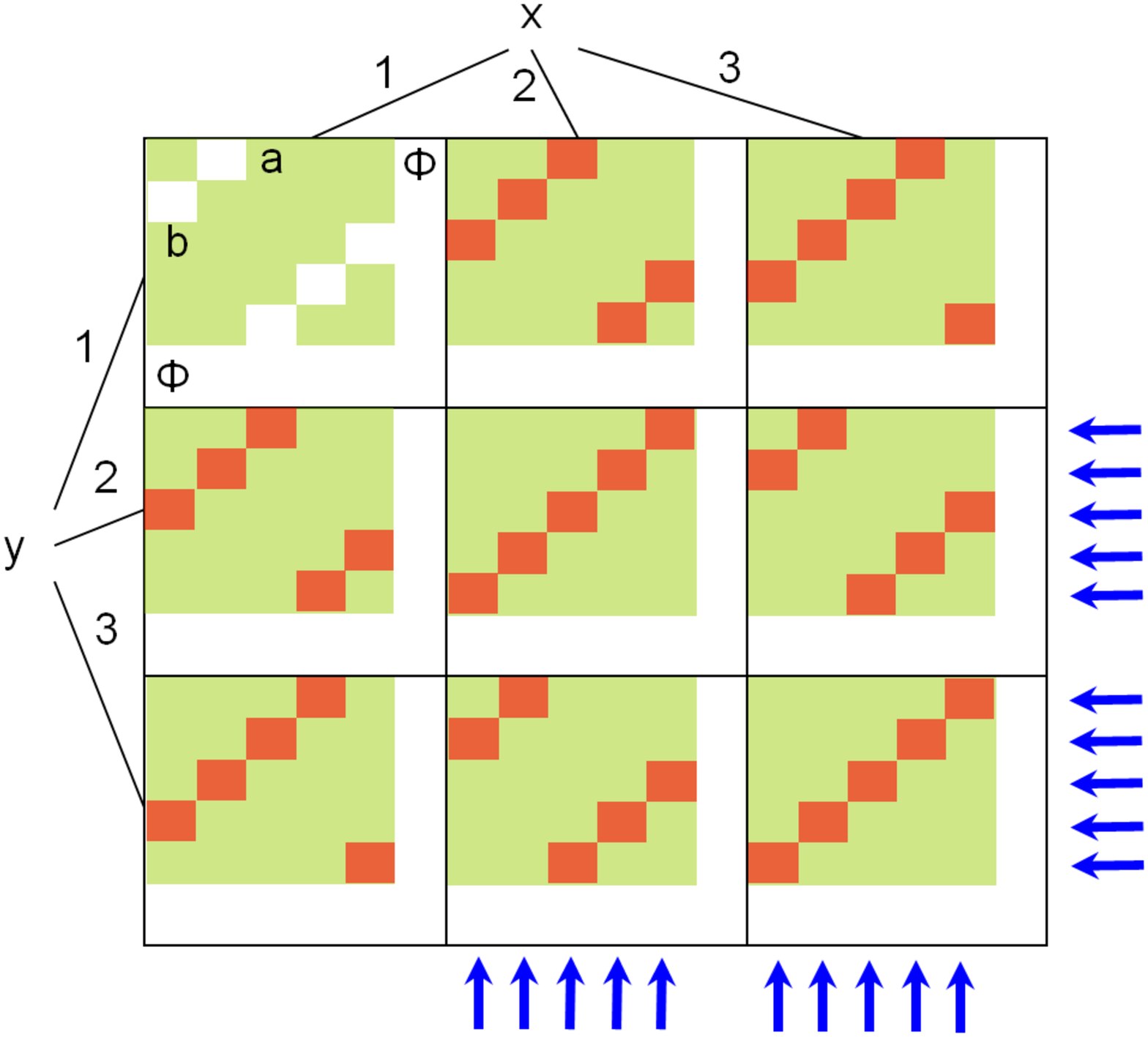}
    }
    \subfigure[]{%
      \includegraphics[width=4cm]{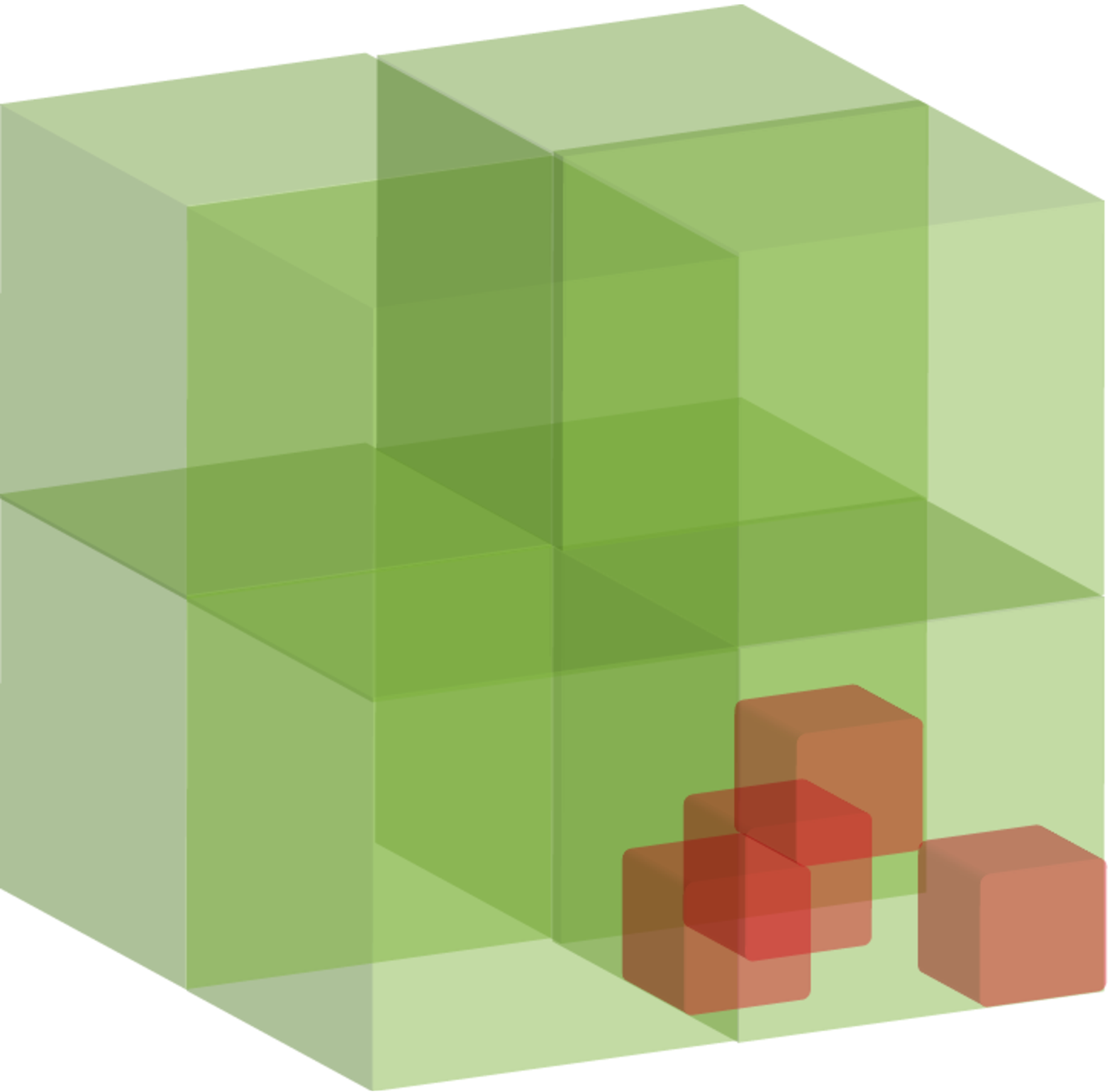}
    }
\caption{Geometric interpretation of the Bell inequalities $I_{M_AM_B(P+1)(P+1)}$ and $I^{\text{multi}}_{P+1}$. (a) In $I_{M_AM_B(P+1)(P+1)}$, which is a linear combination of $p(a,b|x,y)$, whose coefficients are first divided into blocks according to $x,y$, and then divided into entries according to $a,b$. The red, green, white entries are 1, $-P^4$ and 0 respectively. Each blue arrow pointing to a column or a row is a term $-p(a|x)$ or $-p(b|y)$ in $I_{M_AM_B(P+1)(P+1)}$.
(b) $I^{\text{multi}}_{P+1}$ is similarly divided into higher-dimensional blocks and entries. The red entries still maintain the property that there is one and only one red entry on every line within a block according to the construction.}
\label{fig:2Bell}
\end{figure}

First, we prove $I_{M_AM_B(P+1)(P+1)} \leq 0$ with LHV strategies. Since LHV strategies can be regarded as the linear combination of deterministic strategies, the LHV strategy can be assumed deterministic,
\begin{equation}
F(x) = a, G(y) = b.
\end{equation}
In other words, the input will determine output by the functions $F$ and $G$.

Case 1: If $F(x)  = \Phi \ \forall x$, then $I_{M_AM_B(P+1)(P+1)} \le 0$. Similarly, if $G(y) = \Phi \ \forall y$, then $I_{M_AM_B(P+1)(P+1)} \le 0$.

Case 2:
If $x_1\neq x_2, y_1\neq y_2$, $a_1 = F(x_1), a_2 = F(x_2), b_1 = G(y_1), b_2 = G(y_2)$ and none of $a_1,a_2,b_1,b_2$ equals $\Phi$, then one of $f(a_1,b_1,x_1,y_1)$, $f(a_1,b_2,x_1,y_2)$,
 $f(a_2,b_1,x_2,y_1)$, $f(a_2,b_2,x_2,y_2)$ must equal $-P^{4}$. Otherwise
\begin{eqnarray}
a_1 + b_1 &\equiv& x_1y_1 \mod{P}, \label{12.1}\\
a_1 + b_2 &\equiv& x_1y_2 \mod{P},  \label{12.2}\\
a_2 + b_1 &\equiv& x_2y_1 \mod{P}, \label{12.3}\\
a_2 + b_2 &\equiv& x_2y_2 \mod{P}.\label{12.4}
\end{eqnarray}
By applying \eqref{12.1}+\eqref{12.4}$-$\eqref{12.2}$-$\eqref{12.3}, we get $(x_1-x_2)(y_1-y_2) \equiv 0 \mod{P}$. Since $P$ is a prime number, and $1\leq x_1 \neq x_2 \le P,1\leq y_1\neq y_2 \le P$, there is a contradiction. Since $P^{4}$ is bigger than the sum of positive coefficients in the Bell inequality, $I_{M_AM_B(P+1)(P+1)}\leq 0$ follows.

Case 3:
If there is only one $x$ such that $F(x) = a \neq \Phi$, the Bell inequality is simplified to
\begin{equation}
\begin{aligned}
&\sum_{y} f(a, G(y), x, y) + \sum_{y} h(G(y),y) + g(a, x)\\
=&\sum_{y\neq 1} (f(a,G(y),x,y) - 1) + f(a,G(1),x,1) +
g(a, x)\\
\leq& f(a,G(1),x,1) + g(a,x)\leq 0.
\end{aligned}
\end{equation}
If there is only one $y$ such that $G(y) = b \neq \Phi$, similarly $I_{M_AM_B(P+1)(P+1)} \leq 0$.

Next, we construct a no-signaling strategy
\begin{equation}
p^{NS}(a,b|x,y) =
\begin{cases}
\frac{1}{P} & a+b\equiv xy\\
0 & \text{otherwise}
\end{cases}.
\end{equation}
This probability distribution satisfies the no-signaling condition because when $x,y,a$ are determined, there is a unique $b$ such that $p^{NS}(a,b|x,y) = \frac{1}{P}$ and therefore
\begin{equation}
\begin{aligned}
\forall y,\quad p^{NS}(a|x,y) = \sum_{b}p^{NS}(a,b|x,y) = \frac{1}{P};
\end{aligned}
\end{equation}
similarly $\forall x,\  p^{NS}(b|x,y)=1/P$.
Its Bell value is
\begin{equation}
\begin{aligned}
&I_{M_AM_B(P+1)(P+1)} \\
=& \sum_{a,b,x,y}f(a,b,x,y)p^{NS}_{\eta}(a,b|x,y) \\
&\quad+ \sum_{a,x}g(a,x)p^{NS}_{\eta}(a|x) + \sum_{b,y}h(b,y)p^{NS}_{\eta}(b|y)\\
=&\sum_{x>1\vee y>1} \eta^{2} - \sum_{x>1} \eta - \sum_{y>1} \eta\\
=& \eta^{2}(M_AM_B-1) - \eta(M_A+M_B-2).
\end{aligned}
\end{equation}
Recall that  $\eta > \eta^*_{NS}$ is necessary to violate the Bell inequality, i.e., $I_{M_AM_B(P+1)(P+1)}>0$.
Hence Eq.~\eqref{eq:twoupper} holds.
\end{proof}

 Through proving Theorem 2, we design a useful Bell inequality $I_{M_AM_B(P+1)(P+1)}$. It has a similar condition with the CHSH inequality: $a\oplus b \equiv x\cdot y$ \cite{clauser1969proposed}, but generalizes module $2$ to module $P$. Compared to another generalization of the CHSH inequality \cite{collins2004relevant}, our Bell inequality is more advantageous at persisting nonlocality when detector inefficiency occurs. There are relatively few quantum efficiency upper bound results. The quantum efficiency upper bound is $61.8\%$ when $M_A=M_B=4$, showing the $2/3$ bound can be beaten with relatively few inputs \cite{vertesi2010closing}.
  With more inputs $M_A=M_B=2^{d}$, the quantum efficiency bound can be as low as $\eta = d^{1/2}2^{-0.0035d}$  \cite{massar2002nonlocality}, which is however hard to be met in practice, requiring $10^{285}$ inputs when the efficiency is $1/10$. Our Bell inequality suggests much fewer inputs might suffice.

We next generalize the efficiency bound $\eta_{NS}^{*}$ to $N$ parties and have the following lower bound.
\begin{theorem}
In Bell tests with $N\le 500$ space-separated parties, the efficiency bound satisfies
\begin{equation}
\eta_{NS}^{*}\ge  \frac{N}{M(N-1)+1},
\end{equation}
 where all $N$ parties have $M$ inputs.
\end{theorem}
\begin{proof}
From Ref.~\cite{PhysRevA.68.062109} Conjecture 2, a multipartite LHV model is designed to simulate any quantum strategy when the efficiency is no larger than $M/[M(N-1)+1]$ and $N\leq 500$. Since the only condition used in that proof is the no-signaling condition, this finishes the proof. Ref.~\cite{PhysRevA.68.062109} also conjectures that this bound holds for any value of $N$.
\end{proof}

When $N\rightarrow \infty$, this theorem suggests $1/M$ is the detector efficiency lower bound  for showing nonlocality. We also have the following asymptotically matched upper bound.

\begin{theorem}
Consider $N$-party Bell tests with $M_1, M_2, ..., M_N$ inputs respectively,
\begin{equation}
\label{multiupper}
\eta_{NS}^{*}\le \left(\frac{M_1+M_2+\cdots+M_N-N}{M_1M_2\cdots M_N-1}\right)^{\frac{1}{N-1}}.
\end{equation}
\end{theorem}
\begin{proof}
Denote $P$ as the minimal prime number such that $P> \max(M_1,M_2,\cdots,M_N)$.
We construct a general Bell inequality with $P$ outputs excluding $\Phi$,
\begin{equation}
\begin{aligned}
I_{P+1}^{\text{multi}}
&= \sum_{\vec{a},\vec{x}} f(\vec{a},\vec{x})p(\vec{a}|\vec{x})
+ \sum_i\sum_{a_i,x_i} g_i(a_i,x_i)p(a_i|x_i) \\
f(\vec{a},\vec{x}) &=
\begin{cases}
1 & (\exists t, x_t \neq 1)\wedge  (\sum_{j}{a_j} \equiv \prod_{i} x_i)\\
0  &	(\forall t, x_t = 1)\wedge (\sum_{i}{a_i} \equiv \prod_{i} x_i)\\
-P^{2N} & \text{otherwise}\\
\end{cases}
\\
g_i(a_i,x_i) &=
\begin{cases}
-1 & x_i > 1\\
0 & x_i = 1
\end{cases},
\end{aligned}
\end{equation}
where $\vec{a} =\{a_1, ..., a_N\}, \vec{x} = \{x_1, ..., x_N\}$, $1\leq x_i\leq M_i, 0\leq a_i<P$ and the modulo is over $P$. An illustration of this Bell inequality is shown in Fig.~\ref{fig:2Bell}(b).

First, we prove $I_{(P+1)}^{\text{multi}} \leq 0$ for any LHV strategy. Since a LHV strategy can always be regarded as a linear combination of deterministic strategies, we consider a deterministic strategy $
F_i(x_i) = a_i\ \forall i$.

Case 1:
If there exists some $i$ such that $F_i(x_i) = \Phi \ \forall 1\leq x_i \le M_i$, then $I_{P+1}^{\text{multi}} \le 0$.

Case 2:
There exist $i,j,x_i',x_i'', x_j', x_j''$, where $x_i',x_i''$ are different inputs of the $i$-th party, $x_j', x_j''$ are different inputs of the $j$-th party,  such that none of the corresponding outputs $a_i'=F_i(x_i'), a_i''=F_i(x_i''), a_j'=F_j(x_j'), a_j''=F_j(x_j'')$ equals $\Phi$.
Also for any $k$, there exists $x_k$ such that $F_k(x_k) = a_k \neq \Phi$ because otherwise Case 1 arises. Denote $\vec{a_1} = \{a_k\}$, $\vec{x_1} = \{x_k\}$ and $\vec{a_1}_{a_i', a_j'}$ as the vector which replaces the $i$-th and the $j$-th items in $\vec{a_1}$ with $a_i', a_j'$ respectively and similarly for $\vec{x_1}_{x_i',x_j'}$. Consider the following four Bell coefficients:
\begin{align*}
f(\vec{a_1}_{a_i',a_j'},\vec{x_1}_{x_i',x_j'}),f(\vec{a_1}_{a_i'',a_j'},\vec{x_1}_{x_i'',x_j'}),\\
f(\vec{a_1}_{a_i',a_j''},\vec{x_1}_{x_i',x_j''}),f(\vec{a_1}_{a_i'',a_j''},\vec{x_1}_{x_i'',x_j''}).
\end{align*}
Denote $c = \sum_{k\neq i, k\neq j} a_k, z = \prod_{k\neq i, k\neq j} x_k$ and consider
 \begin{eqnarray}
a_i' + a_j' +c&\equiv x_i'x_j' z \mod{P}  \label{(1)},\\
a_i' + a_j''  +c&\equiv x_i'x_j'' z \mod{P}\label{(2)},\\
a_i'' + a_j'  +c&\equiv x_i''x_j' z \mod{P}\label{(3)},\\
a_i'' + a_j''  +c&\equiv x_i''x_j'' z \mod{P}\label{(4)}.
\end{eqnarray}
If one of the above equations does not hold, then $-P^{2N}$ will contribute to the Bell inequality which induces $I^{\text{multi}}_{P+1} \leq 0$, because the sum of the positive coefficients is smaller than $P^{2N}$.
If all equations hold, $\eqref{(1)}+\eqref{(4)}-\eqref{(2)}-\eqref{(3)}$ yields $(x_i'-x_i'')(x_j'-x_j'')z \equiv 0$, which contradicts with that $P$ is a prime number and $1\leq x_i',x_i'',x_j',x_j'',x_k< P$.

Case 3:
Without loss of generality, we assume for $1\leq i\leq N-1$, there is only one $x_i'$ such that $F_i(x_i') = a_i' \neq \Phi$. Denote $\vec{a'} = \{a_i'\}_{1 \leq i \leq N-1}$, $\vec{x'} = \{x_i'\}_{1 \leq i \leq N-1}$. Then $I^{\text{multi}}_{P+1}=\sum_{x_n} f(\vec{a'},F_n(x_n),\vec{x'},x_n)+\sum_{i=1}^{N-1} g_i(a_i',x_i') + \sum_{x_n} g_n(F_n(x_n), x_n) =\sum_{x_n\neq 1} (f(\vec{a'},F_n(x_n),\vec{x'},x_n) - 1)
+ \sum_{i=1}^{N-1} g_i(a_i',x_i') + f(\vec{a'}, F_n(1), \vec{x'}, 1)
\leq \sum_{i=1}^{N-1} g_i(a_i',x_i') + f(\vec{a'}, F_n(1), \vec{x'}, 1)
\leq 0$. The last inequality holds because if $f(\vec{a'}, F_n(1), \vec{x'}, 1)=1$, then there exists $1\le j<n$ such that $x_j'>1$ and consequently $g_j(a_j',x_j')=-1$.

Next, we construct a no-signaling strategy
\begin{equation}
\begin{aligned}
p^{NS}(\vec{a}|\vec{x})= \begin{cases}
\frac{1}{P^{N-1}} & \sum_{i=1}^{N}a_i\equiv \prod_{i=1}^{N}x_i\\
0 & \text{otherwise}
\end{cases},
\end{aligned}
\end{equation}
which satisfies the no-signaling condition
\begin{equation}
\begin{aligned}
\forall \vec{a},\vec{x},1\le k\le N \quad  p^{NS}(a_k|\vec{x}) = p^{NS}(a_k|x_k),
\end{aligned}
\end{equation}
because for any value of $a_k$,
\begin{equation}
\begin{aligned}
p^{NS}(a_k|\vec{x}) =\sum_{\vec{a'}:a'_k = a_k, \sum_{i}a'_i \equiv \prod_{i}x_i} \frac{1}{P^{N-1}} =\frac{1}{P}.
\end{aligned}
\end{equation}
The last equality holds because there are only $N-2$ free variables.

The Bell value of the no-signaling strategy $p^{NS}_{\eta}$ is
\begin{equation}
\begin{aligned}
I_{P+1}^{\text{multi}}=& \sum_{\vec{a},\vec{x}} f(\vec{a},\vec{x})p(\vec{a}|\vec{x}) +\sum_i \sum_{a_i,x_i}g_i(a_i, x_i)\\
=&\sum_{\vec{x}:\exists t, x_t \neq 1} \eta^{N} - \sum_i \sum_{x_i>1} \eta \\
=&\eta^N(\prod_{i=1}^nM_i-1)- \eta(\sum_{i=1}^nM_i-N).
\end{aligned}
\end{equation}
Combined with the definition of $\eta^*_{NS}$ that $\eta > \eta^*_{NS}$ is necessary to violate the Bell inequality, i.e., $I_{P+1}^{\text{multi}}>0$, this leads to Eq.~\eqref{multiupper}.
\end{proof}

We compare the efficiency lower bound and upper bound in Theorem 3 and Theorem 4 for different numbers of parties $N$ in Fig.~\ref{fig:upperlower}. In the comparison, the number of inputs for each party is taken to be the same value $M$. It can be seen that when $N=2$, the two bounds coincides. Also when $N$ becomes large, the two bounds converge to the same value. This is consistent with the analytic upper and lower bound formulas, which both converge to $1/M$ when $N$ goes to infinity.

\begin{figure}
\centering
\includegraphics[width=9cm]{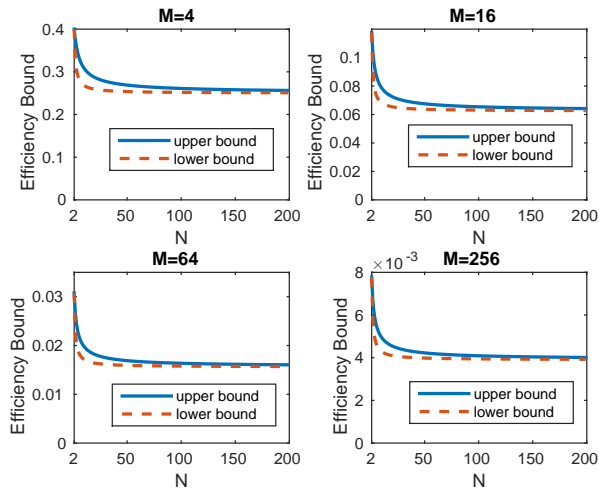}
\caption{(a)-(d) varies $M$ from 4 to 256 and for each $M$, the efficiency upper bound and lower bound are shown for $N=2,3,\cdots,200$. When $M$ increases, the ratio between the efficiency bound of $N=2$ and $N=200$ gradually increases but is always smaller than 2.}
\label{fig:upperlower}
\end{figure}

It is also instructive to fix $M$ and compare the efficiency bound for different $N$. It can be seen that with increasing $N$, the efficiency bound is lowered. However, even when $N$ goes to infinity, the efficiency bound is lowered by at most half compared to $N=2$. On the other hand, increasing the number of input settings $M$ greatly reduces the efficiency requirement. For example, when $M=256$ and $N=2$, an  efficiency  of $8\%$ is enough to demonstrate nonlocality. Thus in this respect, increasing the number of input settings is more efficient than increasing the number of parties.

Compared to quantum efficiency bounds, there exist significant gaps. When $N=3$ and $M=8$, the best bound of the quantum theory is $0.501$ \cite{pal2015closing} while the best bound of the maximally nonlocal theory is $0.20$. For $N\to \infty$, the best bound of the quantum theory $2/(2+M)$ \cite{PhysRevA.92.052104} is approximately twice of the best bound of the maximally nonlocal theory $1/M$. Therefore the advantage brought by our no-signaling strategies may provide inspiration to design more robust quantum strategies.

\section{Discussion}
In summary, we have investigated the efficiency requirement for the violation of Bell inequalities in the maximally nonlocal theory. Our result implies that, for showing the maximally nonlocal theory, the detector efficiency requirement can be arbitrarily low with enough input settings. Our work opens a few interesting avenues for future investigations. First, though our work essentially closes the gap between the upper bound and the lower bound of the detector efficiency for the maximally nonlocal theory, the corresponding question in the quantum theory is still wide open. It would be interesting to apply our techniques to solve the analog question in the quantum theory.
Second, there is a small gap between the detector efficiency bounds of the maximally nonlocal theory when $2<N<\infty$. Closing this gap may require some more delicate estimates on both the upper bound and  the lower bound.

\section*{Acknowledgements}The authors acknowledge insightful discussions with X. Ma. This work was supported by the 1000 Youth Fellowship program in China.

Z. C. and T. P. contributed equally to this work.
\appendix
\section{Monotonicity between the efficiency and the simulation capability of the LHV model}
\label{app}
The following lemma shows monotonicity  between the efficiency $\eta$ and the simulation capability of LHV models:

\begin{lemma}\label{monotonicity}
For $\eta_1 < \eta_2$, if $p^{NS}_{\eta_2}$ can be simulated by a LHV model, then $p^{NS}_{\eta_1}$ can also be simulated by a LHV model.
\end{lemma}
\begin{proof}
Assume the LHV model $p_2$ which can simulate $p^{NS}_{\eta_2}$ is
\begin{equation}
\begin{aligned}
p^{NS}_{\eta_2}(a,b|x,y) &= \int_\lambda p_2(\lambda)p_2(a|x,\lambda)p_2(b|y,\lambda)d\lambda\\
p^{NS}_{\eta_2}(a,\Phi|x,y) &= \int_\lambda p_2(\lambda)p_2(a|x,\lambda)p_2(\Phi|y,\lambda)d\lambda\\
p^{NS}_{\eta_2}(\Phi,b|x,y) &= \int_\lambda p_2(\lambda)p_2(\Phi|x,\lambda)p_2(b|y,\lambda)d\lambda\\
p^{NS}_{\eta_2}(\Phi,\Phi|x,y) &= \int_\lambda p_2(\lambda)p_2(\Phi|x,\lambda)p_2(\Phi|y,\lambda)d\lambda\\
\end{aligned}
\end{equation}
Then, we define a LHV model $p_1$ such that
\begin{equation}
\begin{aligned}
p_1(a|x,\lambda) &= p_2(a|x,\lambda) \cdot \frac{\eta_1}{\eta_2}\\
p_1(\Phi|x,\lambda) &= 1-\frac{\eta_1}{\eta_2} + p_2(\Phi|x, \lambda) \cdot \frac{\eta_1}{\eta_2}\\
p_1(b|y,\lambda) &= p_2(b|y,\lambda) \cdot \frac{\eta_1}{\eta_2}\\
p_1(\Phi|y,\lambda) &= 1-\frac{\eta_1}{\eta_2} + p_2(\Phi|y, \lambda) \cdot \frac{\eta_1}{\eta_2}\\
\end{aligned}
\end{equation}
It is easy to verify that $p_1$ simulates $p_{\eta_1}^{NS}$. Intuitively, we modify the strategy $p_2$ to $p_1$ by assuming the detector efficiency of $p_2$ to be $\eta_1/\eta_2$: outputting $\Phi$ (empty) with a probability $1-\eta_1/\eta_2$.
\end{proof}

By Lemma \ref{monotonicity}, we denote the function $e:p^{NS}\rightarrow [0,1]$, such that $e(p^{NS})$ is the maximal efficiency which satisfies that $p^{NS}_{e(p^{NS})}$ can be simulated by a LHV model.

Thus, the efficiency bound $\eta^*_{NS}$ of the no-signaling strategy can be defined as follows (with the number of inputs fixed):
\begin{equation}
\eta^{*}_{NS} = \inf_{p^{NS}} e(p^{NS}).
\end{equation}

\bibliographystyle{apsrev4-1}

\bibliography{BibBell}

\end{document}